\documentclass[journal,twocolumn]{IEEEtran}
\ifCLASSINFOpdf
\else
\fi
\usepackage[noadjust]{cite}
\usepackage{amsmath}
\usepackage{graphicx}
\usepackage{amssymb}
\usepackage{colortbl}
\usepackage{arydshln}
\usepackage{stfloats}
\usepackage[vlined,boxed,ruled,linesnumbered]{algorithm2e}
\SetKwProg{Fn}{Function}{}{}
\usepackage{algpseudocode}
\usepackage{multirow}
\usepackage{extarrows}
\usepackage{mathrsfs}
\usepackage{upgreek}
\newtheorem{corollary}{Corollary}
\newtheorem{problem}{Problem}
\newtheorem{example}{Example}
\newtheorem{assumption}{Assumption}
\newtheorem{definition}{Definition}
\newtheorem{theorem}{Theorem}
\newtheorem{lemma}{Lemma}
\newtheorem{remark}{Remark}
\newtheorem{proposition}{Proposition}

\newcommand{\weiming}[1]{{\textcolor{black}{#1}}}

\begin{document}
%
\title{Run-Time Safety Monitoring of Neural-Network-Enabled Dynamical Systems}
%
%
%

\author{Weiming Xiang,~\IEEEmembership{Senior Member,~IEEE}
\thanks{Weiming Xiang is with the School of Computer and Cyber Sciences, Augusta University, Augusta,
GA, 30912 USA e-mail: wxiang@augusta.edu.}
}

%
%

\markboth{IEEE Transactions on Cybernetics,~Vol.XX, No.XX, January~2021}%
{Shell \MakeLowercase{\textit{et al.}}: Bare Demo of IEEEtran.cls for IEEE Journals}
%



\maketitle

\begin{abstract}
Complex dynamical systems rely on the correct deployment and operation of numerous components, with state-of-the-art methods relying on learning-enabled components in various stages of modeling, sensing, and control at both offline and online levels. This paper addresses the run-time safety monitoring problem of dynamical systems embedded with neural network components. A run-time safety state estimator in the form of an interval observer is developed to construct lower-bound and upper-bound of system state trajectories in run time. The developed run-time safety state estimator consists of two auxiliary neural networks derived from the neural network embedded in dynamical systems, and observer gains to ensure the positivity, namely the ability of estimator to bound the system state in run time, and the convergence of the corresponding error dynamics. The design procedure is formulated in terms of a family of linear programming feasibility problems. The developed method is illustrated by a numerical example and is validated with evaluations on an adaptive cruise control system.  
\end{abstract}


%
\IEEEpeerreviewmaketitle

\section{Introduction}
Complex dynamical systems,  for instance,  medical robotic systems,  autonomous vehicles and a  variety of cyber-physical systems (CPS), have been increasingly benefiting from the recent rapid development of machine learning (ML) and artificial intelligence (AI) techniques in various aspects ranging from modeling to control, for instance stabilizing neural network controllers and state observers \cite{wu2014exponential,li2019neural,zhang2016state}, adaptive neural network controllers \cite{niu2017command,niu2019multiple} and a variety of neural network controllers \cite{hunt1992neural}. However, because of the well-known vulnerability of neural networks, those systems equipped with neural networks which are also called neural-network-enabled systems are only restricted to scenarios with the lowest levels of the requirement of safety. As often observed, a slight perturbation that is imposed onto the input of a well-trained neural network would lead to a completely incorrect and unpredictable result \cite{szegedy2013intriguing}. When neural network components are involved in dynamical system models such as neural network controllers applied in feedback channels, there inevitably exist noises and disturbances in output measurements of the system that are fed into the neural network controllers. These undesired but inevitable noises and disturbances may bring significant safety issues to dynamical systems in run-time operation. Moreover, with advanced adversarial machine learning techniques recently developed which can easily attack learning-enabled systems in run time, the safety issue of such systems only becomes worse. Therefore, for the purpose of safety assurance of dynamical systems equipped with neural network components, there is a need to develop safety monitoring techniques that are able to provide us the online information regarding safety properties for neural-network-enabled dynamical systems.  

To assure the safety property of neural networks, there are a few safety verification methods developed recently. These approaches are mostly designed in the framework of offline computation and usually represent high computational complexities and require huge computation resources to conduct safety verification. For instance, the verification problem of a class of neural networks with rectified linear unit (ReLU) activation functions can be formalized as a variety of sophisticated computational problems.  One geometric computational approach based on the manipulation of polytopes is proposed in \cite{xiang2018reachable,xiang2017reachable}  which is able to compute the exact output set of an ReLU neural network. In their latest work \cite{tran2019star,tran2019parallelizable}, a novel Star set is developed to significantly improve the scalability. Optimization-based methods are also developed for verification of ReLU neural networks such as mixed-integer linear programming  (MILP)  approach \cite{lomuscio2017approach,dutta2018output}, linear programming (LP) based approach \cite{ehlers2017formal}, and Reluplex algorithm proposed in \cite{katz2017reluplex} which is stemmed from classical Simplex algorithm. For neural networks with general activation function, a  simulation-based approach is introduced in \cite{xiang2018output} inspired by the maximal sensitivity concept proposed in  \cite{zeng2001sensitivity}. The output reachable set estimation for feedforward neural networks with general activation functions is formulated in terms of a chain of convex optimization problems, and an improved version of the simulation-based approach is developed in the framework of interval arithmetic \cite{xiang2020reachable,xiang2018specification}. These optimization and geometric methods require a substantial computational ability to verify even a simple property of a neural network. For example, some properties in the proposed ACAS Xu neural network in \cite{katz2017reluplex} need even more than 100 hours to complete the verification, which does not meet the real-time requirement of run-time safety monitoring for dynamical systems. 

One way to resolve the real-time challenge of run-time monitoring is to develop more computational efficient verification methods that can be executed sufficiently fast to satisfy the run-time requirement such as the specification-guide method and Star set method do in \cite{xiang2018specification,tran2019star,tran2019parallelizable}. However, these offline methods are essentially with an open-loop computation structure and there always exist computational limitations for these offline algorithms implemented online. On the other hand, inspired by observer design techniques in classical control theory, another way is to design a closed-loop structure of run-time monitoring using the instantaneous measurement of the system, which is illustrated in Figure \ref{fig:runtime}. Recently, interval observer design techniques have been developed to provide lower- and upper-bounds of state trajectories during the system's operation which can be used to conduct run-time monitoring for dynamical systems \cite{gouze2000interval,efimov2016design,briat2016interval,chen2017l1,chen2018state,li2019interval,tang2019interval}. Inspired by the idea of
interval observer methods developed in the framework of positive systems \cite{chen2019reachable,xiang2017stability,shen2016static,briat2013robust}, a novel run-time safety state estimator is developed for neural-network-enabled dynamical systems. The run-time state estimator design consists of two essential elements, the auxiliary neural networks and observer gains. Briefly speaking, the auxiliary neural networks stemmed from the neural network in the original system are designed to deal with neural network components and observer gains are handling system dynamics, ensuring positivity of error states and convergence.  The design process can be formulated in terms of a family of LP feasibility problems. Notably, if the neural network component is driven by measurement instead of system state, the design process is independent with the neural network which makes the developed method applicable for neural-network-enabled systems regardless of the scalability concern for the size of neural networks. 
\begin{figure}
\centering
	\includegraphics[width=7.5cm]{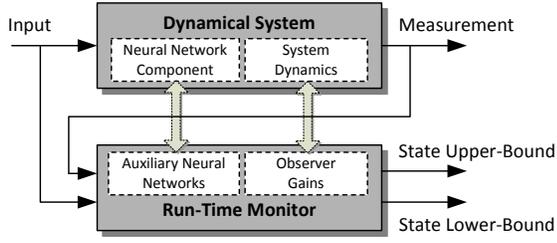}
	\caption{The generic structure of run-time safety monitoring of neural-network-enabled dynamical systems considered in this paper.}\label{fig:runtime}
\end{figure}

The rest of this paper is organized as follows. In Section II, some preliminaries 
and problem formulation are introduced.
The main result, run-time monitoring design, is proposed in Section III. Two auxiliary neural networks are derived from the weights and biases of the neural network of original systems. Interval observers are designed in the framework of LP problems and furthermore, the convergence of the error system is discussed. 
In Section IV, the developed approach is applied to an adaptive cruise control (ACC) system. 
Conclusions and future remarks are given in Section V. 

\emph{Notations:}  $\mathbb{R}$ and $\mathbb{R}_+$  stand for the sets of real numbers and nonnegative real numbers respectively, and $\mathbb{R}^n$ denotes  the vector space of all $n$-tuples of real numbers, $\mathbb{R}^{n\times n}$  is the space of $n\times n$  matrices with real entries. We denote $I_{n \times n} \in \mathbb{R}^{n \times n}$ as an $n$-dimensional identity matrix and $\mathbf{1}_{n \times 1}=[1,\ldots,1]^{\top} \in \mathbb{R}^{n \times 1}$.
Matrix $A \in \mathbb{R}^{n\times n}$ is a Metzler matrix if its off-diagonal entries are nonnegative, and $\mathbb{M}_n$ denotes the set of the Metzler matrices of the size $n$. For $x \in \mathbb{R}^n$, $x_i$  denotes the $i$th component of $x$, and the notation $x> 0$ means $x_i>0$ for  $1 \le i \le n$.  $\mathbb{R}_ + ^n = \{ x \in {\mathbb{R}^n}:x> 0\} $ denotes the nonnegative orthant in $\mathbb{R}^n$, $\mathbb{R}_ + ^{n \times m}$ denotes the set of $n \times m$ real non-negative matrices. For $x \in \mathbb{R}^n$, its 1-norm is $\left\| x\right\| = \sum\nolimits_{k = 1}^n {\left| {{x_k}} \right|} $. Similarly, for an $A \in \mathbb{R}^{n\times m}$,  $a_{ij}$ denotes the element in the $(i,j)$  position of $A$, and  $A> 0$  means that $a_{ij}>0$  for $1 \le i,j \le n$. $A> B$ means that $A-B> 0$. $\left|A\right|$ means $\left|a_{ij}\right|$ for $1 \le i,j\le n$ and $A^{\top}$ is the transpose of $A$.

\section{System Description and Problem Formulation}

\subsection{Neural-Network-Enabled Dynamical Systems}

In this work, we consider an $L$-layer feedforward neural network $\Phi: \mathbb{R}^{n_0} \to \mathbb{R}^{n_{L}}$ defined by the following recursive equations in the form of 
\begin{align}\label{eq:nn_0}
\begin{cases}
    \eta_{\ell} = \phi_{\ell}(W_{\ell} \eta_{\ell-1}+b_\ell),~\ell = 1,\ldots,L
    \\
    \Phi(\eta_0) =\eta_{L}
    \end{cases}
\end{align}
where $\eta_\ell$ denotes the output of the $\ell$-th layer of the neural network, and in particular $\eta_0\in\mathbb{R}^{n_0}$ is the input to the neural network and $\eta_L\in \mathbb{R}^{n_L}$ is the output produced by the neural network, respectively. $W_\ell \in \mathbb{R}^{n_{\ell}\times n_{\ell-1}}$ and $b_{\ell} \in \mathbb{R}^{n_{\ell}}$ are weight matrices and bias vectors for the $\ell$-th layer.  $\phi_\ell = [\psi_{\ell},\cdots,\psi_{\ell}]$ is the concatenation of activation functions of the $\ell$-th layer in which $\psi_{\ell}:\mathbb{R} \to \mathbb{R}$ is the activation function.  The following assumptions which are related to activation functions are proposed. 
\begin{assumption}\label{assumption_1}
Assume that the following properties holds for activation functions $\psi_\ell$, $\ell = 1,\ldots,L$:
\begin{enumerate}
    \item Given any two scalar $x_1$ and $x_2$, there exists an $\alpha > 0$ such that 
    \begin{equation}\label{eq:assumption_1_1}
        \left|\psi_\ell(x_1) - \psi_\ell(x_2)\right| \le \alpha \left|x_1-x_2\right|,~\forall \ell = 1,\ldots,L .
    \end{equation}
    \item Given any two scalars $x_1 \le x_2$, the following inequality holds
\begin{equation}\label{eq:assumption_1_2}
    \psi_\ell(x_1) \le \psi_{\ell}(x_2),~\forall \ell = 1,\ldots,L .
\end{equation}
\end{enumerate}
\end{assumption}
\begin{remark}
The above two assumptions hold for most popular activation functions such as ReLU, sigmoid, tanh, for instance.  The maximum Lipschitz constant of all $\psi_\ell$ can be chosen as the $\alpha$ for condition (\ref{eq:assumption_1_1}). 
In addition, those popular activation functions are monotonically increasing so that  condition (\ref{eq:assumption_1_2}) is explicitly satisfied. 
\end{remark}

Neural-network-enabled dynamical systems are dynamical systems driven by neural network components such as neural network feedback controllers. In general,  neural-network-enabled dynamical systems are in the form of
\begin{align}\label{eq:system}
\begin{cases}
    \dot x(t)  = f(x(t),u(t),\Phi(x(t),u(t)))
    \\
    y(t) = g(x(t),u(t))
\end{cases}
\end{align}
where $x(t) \in \mathbb{R}^{n_x}$ is the state vector, $u(t) \in \mathbb{R}^{n_u}$ is the system input and $y(t) \in \mathbb{R}^{n_y}$ is the measurement of the system, respectively. $f:\mathbb{R}^{n_x+n_u} \to \mathbb{R}^{n_x}$ and $g:\mathbb{R}^{n_x+n_u} \to \mathbb{R}^{n_y}$ are nonlinear functions.  $\Phi:\mathbb{R}^{n_x+n_u} \to \mathbb{R}^{n_x}$ is the neural network component embedded in the system dynamics. In the rest of this paper, the time index $t$ in some variables may be
omitted for brevity if no ambiguity is introduced. 

In the work, we focus on a class of neural-network-enabled systems with system dynamics in the form of Lipschitz nonlinear model described as
\begin{align}\label{eq:Lipsystem}
\mathfrak{L}:
\begin{cases}
    \dot x  = Ax + f(x) + \Phi(x,u)
    \\
    y = Cx
\end{cases}
\end{align}
where $A \in \mathbb{R}^{n_x \times x_x}$, $C \in \mathbb{R}^{n_y \times n_x}$, and $f(x,u)$ is a Lipschitz nonlinearity satisfying the Lipschitz inequality 
\begin{align} 
    \left\|f(x_1)-f(x_2)\right\| \le \beta \left\|x_1-x_2\right\|,~ \beta > 0.
\end{align}
\begin{remark}
It is worthwhile mentioning that any nonlinear system in the form of $\dot x =
f(x)+\Phi(x,u)$ can be expressed in the form of (\ref{eq:Lipsystem}), as long as $f(x)$
is differentiable with respect to $x$. The neural network $\Phi(x,u)$ is an interval component affecting the system behavior. For example, if $\Phi(x,u)$ is trained as the neural network feedback controller, model (\ref{eq:Lipsystem}) represents a state-feedback closed-loop system. 
\end{remark}

Finally, the nonlinearity $f(x)$ is assumed to have the following property.
\begin{assumption}\label{assumption_2}
It is assumed that there exist functions $\underline{f},\overline{f}:\mathbb{R}^{2n_x} \to \mathbb{R}^{n_x}$ such that 
\begin{align}
    \underline{f}(\underline{x},\overline{x}) \le f(x) \le \overline{f}(\underline{x},\overline{x})
\end{align}
holds for any $\underline{x} \le x \le \overline{x}$. 
\end{assumption}

\subsection{Problem Statement}
The run-time safety monitoring problem considered in this paper is to design a run-time safety state estimator $\mathfrak{E}$ which is able to estimate the lower- and upper-bounds of the instantaneous value of $x(t)$ for the purpose of safety monitoring. The information of system $\mathfrak{L}$ that is available for estimator $\mathfrak{E}$ includes:  System matrices $A$, $C$, nonlinearity $f$ and neural network $\Phi$, namely the weight matrices $\{W_\ell\}_{\ell =1}^{L}$ and bias vectors $\{b_\ell\}_{\ell =1}^{L}$, and the $\underline{u}$ and $\overline{u}$ such that input $u(t)$ satisfies $\underline{u} \le u(t)\le \overline{u}$, $\forall t \ge 0$, and the instantaneous value of measurement $y(t)$ in running time.
The run-time safety monitoring problem for neural-network-enabled dynamical system (\ref{eq:Lipsystem}) is summarized as follows.
\begin{problem} \label{problem_1}
Given a neural-network-enabled dynamical system $\mathfrak{L}$ in the form of (\ref{eq:Lipsystem}) with input $u(t)$ satisfying $\underline{u} \le u(t)\le \overline{u}$, $\forall t \ge 0$, how does one design a run-time safety state estimator $\mathfrak{E}$ to reconstruct two instantaneous values $\underline{x}(t)$ and $\overline{x}(t)$ such that $\underline{x}(t) \le x(t) \le \overline{x}(t)$, $\forall t \ge 0$?
\end{problem}

Inspired by interval observers proposed in \cite{gouze2000interval,efimov2016design,briat2016interval,chen2017l1,chen2018state,li2019interval,tang2019interval}, the run-time safety state estimator $\mathfrak{E}$ is developed in the following  Luenberger observer form
\begin{align}
\label{eq:observer}
\mathfrak{E}:
\begin{cases}
    \dot{\underline{x}}  = A\underline{x} + \underline{f}(\underline{x},\overline{x}) + \underline{\Phi}(\underline{x},\overline{x},\underline{u},\overline{u})+\underline{L}(y - C\underline{x})
    \\
    \dot{\overline{x}}  = A\overline{x} + \overline{f}(\underline{x},\overline{x}) + \overline{\Phi}(\underline{x},\overline{x},\underline{u},\overline{u})+\overline{L}(y - C\overline{x})
\end{cases}
\end{align}
where initial states satisfy $\underline{x}(t_0) \le x(t_0) \le \overline{x}(t_0)$ and $\underline{f}$, $\overline{f}$ are functions satisfying Assumption \ref{assumption_2}. Neural networks $\underline{\Phi}$, $\overline{\Phi}$ and observer gains $\underline{L}$, $\overline{L}$ are to be determined. 
Furthermore, letting the error states $\underline{e}(t) = x(t) - \underline{x}(t)$ and $\overline{e}(t) =  \overline{x}(t)-x(t)$, the error dynamics can be obtained as follows:
\begin{align}
\label{eq:errorsystem}
\begin{cases}
    \dot{\underline{e}}  = (A-\underline{L}C)\underline{e} + f(x) - \underline{f}(\underline{x},\overline{x}) + \Phi(x,u)-\underline{\Phi}(\underline{x},\overline{x},\underline{u},\overline{u})
    \\
    \dot{\overline{e}}  = (A-\overline{L}C)\overline{e} + \overline{f}(\underline{x},\overline{x}) - f(x) + \overline{\Phi}(\underline{x},\overline{x},\underline{u},\overline{u})-\Phi(x,u)
\end{cases}
\end{align}
with initial states $\underline{e}(t_0) \ge 0$ and $\overline{e}(t_0) \ge 0$.

The problem of ensuring the run-time value of $x(t)$ satisfying $\underline{x}(t) \le x(t) \le \overline{x}(t)$, $\forall t \ge 0$ is equivalent to the one that the run-time values of error states $\underline{e}(t)$ and $\overline{e}(t)$ are required to be always positive, that is to say, $\underline{e}(t) \ge 0$ and $\overline{e}(t) \ge 0$, $\forall t \ge 0$. 
Thus, with the run-time safety state estimator in the form of (\ref{eq:observer}), the run-time safety monitoring problem for system (\ref{eq:Lipsystem}) can be restated as follows.
\begin{problem}\label{problem_2}
Given a neural-network-enabled dynamical system $\mathfrak{L}$ in the form of (\ref{eq:Lipsystem}) with input $u(t)$ satisfying $\underline{u} \le u(t)\le \overline{u}$, $\forall t \ge 0$, how does one construct proper neural networks $\underline{\Phi}$, $\overline{\Phi}$ and observer gains $\underline{L}$, $\overline{L}$ such that the error states $\underline{e}(t)$, $\overline{e}(t)$ governed by (\ref{eq:errorsystem}) satisfy $\underline{e}(t) \ge 0$ and $\overline{e}(t) \ge 0$, $\forall t \ge 0$? 
\end{problem}

As stated in Problem \ref{problem_2}, the run-time safety monitoring consists of two essential design tasks, neural network design and observer gain design.   
Then, the result about positivity of dynamical systems is recalled by the following lemma. 
\begin{lemma} \cite{efimov2016design}
\label{lemma_1}
Consider a system $\dot z = Mz+p(t)$, $z\in \mathbb{R}^n$, where $M\in \mathbb{M}_{n}$ and $p:\mathbb{R}_+ \to \mathbb{R}^{n}_+$, the system is called cooperative and the solutions of the system satisfy $z(t) \ge 0$, $\forall t \ge 0$ if $z(0) \ge 0$.
\end{lemma}

Based on Lemma \ref{lemma_1} and owing to Assumption \ref{assumption_2} implying $f(x) - \underline{f}(\underline{x},\overline{x}) \in \mathbb{R}_+^{n_x}$ and $\overline{f}(\underline{x},\overline{x}) -f(x) \in \mathbb{R}_+^{n_x}$, the run-time safety monitoring problem  can be resolved if observer gains $\underline{L}$, $\overline{L}$ and nerual networks $\underline{\Phi}$, $\overline{\Phi}$ satisfy the conditions proposed in the following proposition. 
\begin{proposition}\label{proposition_1}
The run-time safety monitoring Problem \ref{problem_1} is solvable if there exist observer gains $\underline{L}$, $\overline{L}$ and neural networks $\underline{\Phi}$ and $\overline{\Phi}$ such that 
\begin{enumerate}
    \item $A-\underline{L}C \in \mathbb{M}_{n_x}$,  $A-\overline{L}C \in \mathbb{M}_{n_x}$.
    \item $\Phi(x,u)-\underline{\Phi}(\underline{x},\overline{x},\underline{u},\overline{u}) \in \mathbb{R}_+^{n_{x}}$, $\overline{\Phi}(\underline{x},\overline{x},\underline{u},\overline{u})-\Phi(x,u) \in \mathbb{R}^{n_{x}}_+$.
\end{enumerate}
\end{proposition}
\begin{proof}Due to Assumption \ref{assumption_2}, it implies $f(x)-\underline{f}(\underline{x},\overline{x}) \in \mathbb{R}_x^{n_x}$. Then, owing to $\Phi(x,u)-\underline{\Phi}(\underline{x},\overline{x},\underline{u},\overline{u}) \in \mathbb{R}_+^{n_{x}}$, one can obtain $f(x)-\underline{f}(\underline{x},\overline{x}) +\Phi(x,u)-\underline{\Phi}(\underline{x},\overline{x},\underline{u},\overline{u}) \in \mathbb{R}_+^{n_{x}}$. Together with $A - \underline{L}C \in \mathbb{M}_{n_x}$, it leads to $\underline{e}(t) \ge 0$, $\forall t \ge 0$ according to Lemma \ref{lemma_1}. Same guidelines can be applied to ensure $\overline{e}(t) \ge 0$. The proof is complete.
\end{proof}

The observer gains $\underline{L}$, $\overline{L}$ and neural networks $\underline{\Phi}$, $\overline{\Phi}$ satisfying conditions in Proposition \ref{proposition_1} can ensure the system state $x(t)$ to be bounded by the estimator states $\underline{x}(t)$ and $\overline{x}(t)$, but there is no guarantee on the boundedness and convergence of error state $\underline{e}(t)$ and $\overline{e}(t)$. The values of $\underline{e}(t)$ and $\overline{e}(t)$ may diverge, namely $\lim_{t \to \infty}\underline{e}(t) = \infty$ and $\lim_{t \to \infty}\overline{e}(t) = \infty$, thus make no sense in terms of safety monitoring in practice.  The following notion of practical stability concerned with boundedness of system state is introduced. 
\begin{definition} \cite{vangipuram1990practical}\label{def_1}
Given $(\epsilon, \delta)$ with $0 < \epsilon \le   \delta$. Let $x(t, x(t_0))$, $t \ge t_0$, be a solution of the system $\dot x(t) = f(x(t), u(t))$, then the trivial
solution $x=0$ of the system is said to be practically stable with respect to $(\epsilon, \delta)$  if $\left\|x(t_0)\right\| \le \epsilon$ implies $\left\|x(t)\right\| \le \delta$, $\forall t \ge t_0$. Furthermore, if there is a
$T = T (t_0,\epsilon, \delta) > 0$ such that $\left\|x(t_0)\right\| \le \epsilon$  implies $\left\|x(t)\right\| \le \delta$ for any $t \ge t_0 + T$, then the system is practically asymptotically stable.
\end{definition}
\begin{remark}
Practical stability ensures the boundedness of state trajectories of a dynamical system. If  the inequality
\begin{align}
    \left\|x(t)\right\| \le Ce^{-\lambda (t-t_0)}\left\|x(t_0)\right\| + r
\end{align}
holds for any $x(t_0) \in\mathbb{R}^{n_x}$, any $t \ge t_0$ and constants $C > 0$, $r \le 0$, it implies that the state trajectories of a  dynamical system are bounded in terms of $\left\|x(t)\right\| \le C\left\|x(t_0)\right\| + r$, $\forall t \ge t_0$, and moreover, the state trajectories also converge to ball $\mathcal{B}_r = \{x \in \mathbb{R}^{n_x} \mid \left\|x\right\| \ge r, r \ge 0\}$  exponentially at a decay rate of $\lambda$.  In particular, we call this system  is globally practically uniformly exponentially  stable.
\end{remark}

\section{Run-Time Safety Monitoring Design}
This section aims to design neural networks $\underline{\Phi}$, $\overline{\Phi}$ and observer gains $\underline{L}$, $\overline{L}$ satisfying conditions in Proposition \ref{proposition_1}. Furthermore, the convergence of error state is also analyzed and assured. 
First, we design neural networks $\underline{\Phi}$ and $\overline{\Phi}$ based on the weight matrices $W_{\ell}$ and bias vectors $b_{\ell}$ of neural network $\Phi$ in system (\ref{eq:Lipsystem}).

Given a neural network $\Phi:\mathbb{R}^{n_0} \to \mathbb{R}^{n_L}$ in the form of (\ref{eq:nn_0}) with weight matrices 
\begin{align}
    W_{\ell} = [w_\ell^{i,j}] = \begin{bmatrix}
    w_{\ell}^{1,1} & w_\ell^{1,2} & \cdots & w_{\ell}^{1,n_{\ell-1}}
    \\
    w_{\ell}^{2,1} & w_\ell^{2,2} & \cdots & w_{\ell}^{2,n_{\ell-1}}
    \\
    \vdots & \vdots & \ddots & \vdots
    \\
    w_{\ell}^{n_\ell,1} & w_\ell^{n_\ell,2} & \cdots & w_{\ell}^{n_\ell,n_{\ell-1}}
    \end{bmatrix}
\end{align}
where $w_\ell^{i,j}$ denotes the element in $i$-th row and $j$-th column,  we define two auxiliary weight matrices as below:
\begin{align} \label{eq:W_1}
    \underline{W}_\ell &= [\underline{w}_\ell^{i,j}],~\underline{w}_\ell^{i,j} = \begin{cases}
        w_{\ell}^{i,j} & w_{\ell}^{i,j} < 0
        \\
        0 & w_{\ell}^{i,j} \ge 0
    \end{cases}
    \\
    \label{eq:W_2}
    \overline{W}_\ell& = [\overline{w}_\ell^{i,j}],~\overline{w}_\ell^{i,j} = \begin{cases}
        w_{\ell}^{i,j} & w_{\ell}^{i,j} \ge 0
        \\
        0 & w_{\ell}^{i,j} < 0
    \end{cases}
\end{align} 
for which it is explicit that we have $W_\ell = \underline{W}_\ell+\overline{W}_\ell$. Then, we construct two auxiliary neural networks $\underline{\Phi}:\mathbb{R}^{2n_0} \to \mathbb{R}^{n_L}$, $\overline{\Phi}:\mathbb{R}^{2n_0} \to \mathbb{R}^{n_L}$ with inputs $\underline{\eta}_0, \overline{\eta}_0 \in \mathbb{R}^{n_0}$ in the following form:
\begin{align}\label{eq:nn_1}
&\begin{cases}
\underline{\eta}_{\ell} = \phi_{\ell}(\underline{W}_{\ell} \overline{\eta}_{\ell-1}+\overline{W}_{\ell} \underline{\eta}_{\ell-1}+b_\ell),~\ell = 1,\ldots,L
    \\
\underline{\Phi}(\underline{\eta}_0,\overline{\eta}_0) =\underline{\eta}_{L}
    \end{cases}
    \\
    \label{eq:nn_2}
&  \begin{cases}
\overline{\eta}_{\ell} = \phi_{\ell}(\underline{W}_{\ell} \underline{\eta}_{\ell-1}+\overline{W}_{\ell} \overline{\eta}_{\ell-1}+b_\ell),~\ell = 1,\ldots,L
    \\
\overline{\Phi}(\underline{\eta}_0,\overline{\eta}_0) =\overline{\eta}_{L}
    \end{cases}
\end{align}

Given $\underline{\eta_0} \le \eta_0 \le \overline{\eta_0}$, the following theorem can be derived with auxiliary neural networks in the form of (\ref{eq:nn_1}) and (\ref{eq:nn_2}), which implies the positivity of $\Phi(\eta_0)-\underline{\Phi}(\underline{\eta_0},\overline{\eta}_0)$ and $\overline{\Phi}(\underline{\eta_0},\overline{\eta}_0)-\Phi(\eta_0)$.

\begin{theorem} \label{thm_1}
Given neural networks $\Phi:\mathbb{R}^{n_0} \to \mathbb{R}^{n_L}$ and its two auxiliary neural networks $\underline{\Phi}:\mathbb{R}^{2n_0} \to \mathbb{R}^{n_L}$, $\overline{\Phi}:\mathbb{R}^{2n_0} \to \mathbb{R}^{n_L}$ defined by (\ref{eq:nn_1}) and (\ref{eq:nn_2}), the following condition 
\begin{align}
    \begin{bmatrix}
   \Phi(\eta_0)-\underline{\Phi}(\underline{\eta_0},\overline{\eta}_0)  \\ \overline{\Phi}(\underline{\eta_0},\overline{\eta}_0)-\Phi(\eta_0)
   \end{bmatrix}\in \mathbb{R}_+^{2n_{L}}
\end{align}
holds for any $\underline{\eta}_0 \le \eta_0 \le \overline{\eta}_0$. 
\end{theorem}

\begin{proof}
Let us consider the $\ell$-th layer. For any $\underline{\eta}_{\ell-1}\le \eta_{\ell-1} \le \overline{\eta}_{\ell-1}$, it can be obtained from (\ref{eq:W_1}), (\ref{eq:W_2}) such that 
\begin{align*}
   \underline{w}_\ell^{i,j}\overline{\eta}_{\ell-1}^{j} + \overline{w}_\ell^{i,j}\underline{\eta}_{\ell-1}^{j} \le w^{i,j}_\ell \eta_{\ell-1}^{j} \le \underline{w}_\ell^{i,j}\underline{\eta}_{\ell-1}^{j}+\overline{w}_\ell^{i,j}\overline{\eta}_{\ell-1}^{j}
\end{align*}
which implies that 
\begin{align*}
  &W_{\ell}\eta_{\ell-1}+b_\ell
  - (\underline{W}_{\ell} \overline{\eta}_{\ell-1}+\overline{W}_{\ell} \underline{\eta}_{\ell-1}+b_\ell) \ge 0
  \\
 & \underline{W}_{\ell} \underline{\eta}_{\ell-1}+\overline{W}_{\ell} \overline{\eta}_{\ell-1}+b_\ell -(W_{\ell}\eta_{\ell-1}+b_\ell)  \ge  0.
\end{align*}

Under Assumption \ref{assumption_1}, the monotonic property (\ref{eq:assumption_1_2})  of activation function $\phi_\ell$ leads to \begin{align*}
     & \phi_\ell(W_{\ell}\eta_{\ell-1}+b_\ell) -\phi_\ell(\underline{W}_{\ell} \overline{\eta}_{\ell-1}+\overline{W}_{\ell} \underline{\eta}_{\ell-1}+b_\ell)\ge 0   \\ &\phi_{\ell}(\underline{W}_{\ell} \underline{\eta}_{\ell-1}+\overline{W}_{\ell} \overline{\eta}_{\ell-1}+b_\ell) - \phi_\ell(W_{\ell}\eta_{\ell-1}+b_\ell) \ge 0 .
\end{align*}

Using the definitions of neural networks $\Phi$, $\underline{\Phi}$ and $\overline{\Phi}$ described by (\ref{eq:nn_0}), (\ref{eq:nn_1}) and (\ref{eq:nn_2}), namely $\eta_{\ell} = \phi_\ell(W_{\ell}\eta_{\ell-1}+b_\ell)$, $\underline{\eta}_{\ell} =\phi_\ell(\underline{W}_{\ell} \overline{\eta}_{\ell-1}+\overline{W}_{\ell} \underline{\eta}_{\ell-1}+b_\ell)$ and $\overline{\eta}_{\ell}=\phi_{\ell}(\underline{W}_{\ell} \underline{\eta}_{\ell-1}+\overline{W}_{\ell} \overline{\eta}_{\ell-1}+b_\ell)$, 
the above derivation implies that
\begin{align}
   \begin{bmatrix}
   \eta_{\ell-1} - \underline{\eta}_{\ell-1}  \\ \overline{\eta}_{\ell-1} -\eta_{\ell-1} 
   \end{bmatrix} \in \mathbb{R}_+^{2n_{\ell-1}} 
   \Rightarrow \begin{bmatrix}
   \eta_{\ell} - \underline{\eta}_{\ell}  \\ \overline{\eta}_{\ell} -\eta_{\ell} 
   \end{bmatrix} \in \mathbb{R}_+^{2n_{\ell}}  .
\end{align}

Thus, given any $\underline{\eta}_{0} \le \eta_{0} \le \overline{\eta}_{0}$, one can  obtain
\begin{align}
   \begin{bmatrix}
   \eta_{L} - \underline{\eta}_{L}  \\ \overline{\eta}_{L} -\eta_{L} 
   \end{bmatrix}
   =
    \begin{bmatrix}
   \Phi(\eta_0)-\underline{\Phi}(\underline{\eta_0},\overline{\eta}_0)  \\ \overline{\Phi}(\underline{\eta_0},\overline{\eta}_0)-\Phi(\eta_0)
   \end{bmatrix}\in \mathbb{R}_+^{2n_{L}} .
\end{align}
The proof is complete.
\end{proof}

With the two auxiliary neural networks $\underline{\Phi}$, $\overline{\Phi}$, we are ready to design observer gains $\underline{L}$ and $\overline{L}$ to construct run-time safety state estimator $\mathfrak{E}$ in the form of (\ref{eq:observer}) via the following theorem.

\begin{theorem}\label{thm_2}
The safety monitoring Problem \ref{problem_1} is solvable if the following conditions hold for observer gains $\underline{L}$, $\overline{L}$ and neural networks $\underline{\Phi}$ and $\overline{\Phi}$:
\begin{enumerate}
    \item There exist $a \in \mathbb{R}$, $\underline{L}, \overline{L} \in \mathbb{R}^{n_x \times n_y}$ such that 
    \begin{align} \label{eq:thm_2_1}
       A - \underline{L}C &\ge aI_{n_x \times n_x}
       \\
       \label{eq:thm_2_2}
        A - \overline{L}C &\ge aI_{n_x \times n_x} .
    \end{align}
    \item $\underline{\Phi}(\underline{x},\overline{x},\underline{u},\overline{u})$, $\overline{\Phi}(\underline{x},\overline{x},\underline{u},\overline{u})$ are in the form of (\ref{eq:nn_1}) and (\ref{eq:nn_2}) with $\underline{\eta}_0 = [\underline{x}^{\top}, \underline{u}^{\top}]^{\top}$ and $\overline{\eta}_0 = [\overline{x}^{\top}, \overline{u}^{\top}]^{\top}$.
\end{enumerate}
\end{theorem}
\begin{proof}
First note that (\ref{eq:thm_2_1}) and (\ref{eq:thm_2_2}) imply that 
    $A-\underline{L}C \in
    \mathbb{M}_{n_x}$, $A-\overline{L}C \in \mathbb{M}_{n_x}$.
Then, from Theorem \ref{thm_1}, it leads to the fact of $\Phi(x,u)-\underline{\Phi}(\underline{x},\overline{x},\underline{u},\overline{u}) \in \mathbb{R}_+^{n_{x}}$, $\overline{\Phi}(\underline{x},\overline{x},\underline{u},\overline{u})-\Phi(x,u) \in \mathbb{R}^{n_{x}}_+$. Based on Proposition \ref{proposition_1}, the error $e(t)$ will be bounded as $\underline{e}(t) \le e(t) \le \overline{e}(t)$, $\forall t \ge 0$, thus the safety monitoring problem is solvable. The proof is complete.
\end{proof}

Theorem \ref{thm_2} provides us a method to design run-time safety state estimator in the interval observer form of (\ref{eq:observer}). The observer gains $\underline{L}$ and $\overline{L}$ can be obtained by solving the linear inequalities (\ref{eq:thm_2_1}), (\ref{eq:thm_2_2}), and neural networks $\underline{\Phi}$ and $\overline{\Phi}$ are determined by (\ref{eq:nn_1}) and (\ref{eq:nn_2}) with weight matrices $\underline{W}_\ell$, $\overline{W}_\ell$ defined by (\ref{eq:W_1}), (\ref{eq:W_2}). The boundedness of $\underline{x}(t) \le x(t) \le \overline{x}(t)$, $\forall t \ge 0$ can be established during the system's operation, however, the boundedness and convergence of error states $\underline{e}(t)$ and $\overline{e}(t)$ are not guaranteed, which means the error dynamics (\ref{eq:errorsystem}) could be unstable. In this case, the estimated bounds $\underline{x}(t)$ and $\overline{x}(t)$ will diverge from system state $x(t)$ to infinite values, and  consequently, the run-time safety monitoring does not make sense in practice. In the following, the convergence of run-time estimation bounds is discussed in the framework of practical stability proposed in Definition \ref{def_1}. 

First, the following assumption is proposed for nonlinearity $f(x)$ and $\underline{f}(\underline{x},\overline{x})$, $\overline{f}(\underline{x},\overline{x})$ mentioned in Assumption \ref{assumption_2}.

\begin{assumption}\label{assumption_3}
It is assumed that there exist scalars $\underline{\gamma}_1$, $\overline{\gamma}_1$, $\underline{\gamma}_2$, $\overline{\gamma}_2 \in \mathbb{R}_+$ and vector $\underline{\rho}, \overline{\rho} \in \mathbb{R}^{n_x}_+$ such that
\begin{align}
    f(x)- \underline{f}(\underline{x},\overline{x})  \le \underline{\gamma}_1(x - \underline{x}) + \underline{\gamma}_2(\overline{x}-x) +\underline{\rho} 
    \\
    \overline{f}(\underline{x},\overline{x}) - f(x) \le \overline{\gamma}_1(x - \underline{x}) + \overline{\gamma}_2(\overline{x}-x)  +\overline{\rho}
\end{align}
holds for $f(x)$, $\underline{f}(\underline{x},\overline{x})$, $\overline{f}(\underline{x},\overline{x})$. 
\end{assumption}

\begin{remark}
These parameters $\underline{\gamma}_1$, $\overline{\gamma}_1$, $\underline{\gamma}_2$, $\overline{\gamma}_2$, $\underline{\rho}$, and  $\overline{\rho}$ in Assumption \ref{assumption_3} can be estimated under Lipschitz condition (6), using the results in \cite{zheng2016design}, i.e., Lemma 6 in \cite{zheng2016design}.
\end{remark}

The following lemma is developed for neural network $\Phi$ and its auxiliary neural networks $\underline{\Phi}$ and $\overline{\Phi}$. 

\begin{lemma}
\label{lemma_2}
Given a feedforward neural network $\Phi:\mathbb{R}^{n_0} \to \mathbb{R}^{n_L}$, there always exist a series of matrices $\underline{S}_\ell, \overline{S}_\ell  \in \mathbb{R}^{n_L\times n_{\ell}}_+$, $\ell = 0,\ldots,L$, with $\underline{S}_L = \overline{S}_L = I_{n_L \times {n_L}}$ such that
\begin{align} \label{eq:lem_2_1}
\alpha\begin{bmatrix}
 \underline{S}_{\ell}\overline{W}_\ell - \overline{S}_{\ell}\underline{W}_{\ell}\\ \overline{S}_\ell\overline{W}_\ell-\underline{S}_{\ell}\underline{W}_{\ell}
\end{bmatrix} &\le 
\begin{bmatrix}
\underline{S}_{\ell-1} \\ \overline{S}_{\ell-1}
\end{bmatrix},~\ell = 1,\ldots,L
    \\  \label{eq:lem_2_2}
    \overline{\Phi}(\underline{\eta}_0,\overline{\eta}_0)-\underline{\Phi}(\underline{\eta}_0,\overline{\eta}_0) & \le     \underline{S}_{0}(\eta_{0}-\underline{\eta}_{0})+\overline{S}_{0}(\overline{\eta}_{0} - \eta_{0})
\end{align}
hold for any $\underline{\eta}_0\le\eta_0 \le \overline{\eta}_0$, where $\alpha$ is the Lipschitz constant of activation functions given in (\ref{eq:assumption_1_1}).
\end{lemma}
\begin{proof}
Starting from $\underline{S}_L = \overline{S}_L = I_{n_L \times {n_L}}$, we can recursively define
\begin{align}
    \underline{S}_{\ell-1} = \alpha( \underline{S}_{\ell}\overline{W}_\ell - \overline{S}_{\ell}\underline{W}_{\ell}) + \epsilon \mathbf{1}_{n_L \times 1} \mathbf{1}_{ n_{\ell-1}\times 1}^{\top} \label{eq:lem_2_3}
    \\
    \overline{S}_{\ell-1} =\alpha(\overline{S}_\ell\overline{W}_\ell-\underline{S}_{\ell}\underline{W}_{\ell}) + \epsilon \mathbf{1}_{n_L \times 1} \mathbf{1}_{ n_{\ell-1} \times 1}^{\top} \label{eq:lem_2_4}
\end{align}
where $\epsilon >0$ could be any positive value. Thus, there always exist $\underline{S}_{\ell}$, $\overline{S}_{\ell}$, $\ell = 0,\ldots,L$ such that (\ref{eq:lem_2_1}) holds. 

Then, we are going to establish (\ref{eq:lem_2_2}). We consider the $\ell$-th layer $ \eta_{\ell} = \phi_{\ell}(W_{\ell} \eta_{\ell-1}+b_\ell)$. 
Under Assumption \ref{assumption_1}, it implies 
\begin{align}
\eta_\ell - \underline{\eta}_\ell 
 =& \phi_{\ell}({W}_{\ell} \eta_{\ell-1}+b_\ell) -\phi_{\ell}(\underline{W}_{\ell} \overline{\eta}_{\ell-1}+\overline{W}_{\ell} \underline{\eta}_{\ell-1}+b_\ell) \nonumber 
\\
\le& \alpha \left|{W}_{\ell} \eta_{\ell-1}+b_\ell-\underline{W}_{\ell} \overline{\eta}_{\ell-1}-\overline{W}_{\ell} \underline{\eta}_{\ell-1}-b_\ell\right| \label{eq:lem_2_5}
\end{align}

Following the same guideline in the proof of Theorem \ref{thm_1}
one obtains 
\begin{align} \label{eq:lem_2_6}
    &{W}_{\ell} \eta_{\ell-1}+b_\ell -\underline{W}_{\ell} \overline{\eta}_{\ell-1}-\overline{W}_{\ell} \underline{\eta}_{\ell-1}-b_\ell \ge 0 
\end{align}
and using the fact of $W_\ell = \underline{W}_\ell+\overline{W}_\ell$, inequality (\ref{eq:lem_2_5}) equals 
\begin{align}
\eta_\ell - \underline{\eta}_\ell
\le& \alpha \overline{W}_{\ell} (\eta_{\ell-1}-\underline{\eta}_{\ell-1}) - \alpha \underline{W}_\ell(\overline{\eta}_{\ell-1}-\eta_{\ell-1}) \nonumber
\\
=& 
\begin{bmatrix}
\alpha \overline{W}_{\ell} & - \alpha \underline{W}_\ell
\end{bmatrix}
\begin{bmatrix}
\eta_{\ell-1}-\underline{\eta}_{\ell-1}
\\
\overline{\eta}_{\ell-1}-\eta_{\ell-1}
\end{bmatrix} . \label{eq:lem_2_8}
\end{align}

Similarly, one can obtain
\begin{align}
\overline{\eta}_\ell - \eta_\ell
\le \begin{bmatrix}
-\alpha \underline{W}_{\ell} & \alpha\overline{W}_{\ell}
\end{bmatrix}
\begin{bmatrix}
\eta_{\ell-1} - \underline{\eta}_{\ell-1}
\\
\overline{\eta}_{\ell-1}-\eta_{\ell-1}
\end{bmatrix} . \label{eq:lem_2_9}
\end{align}

Based on inequalities
(\ref{eq:lem_2_8}) and (\ref{eq:lem_2_9}), the following inequality can be established
\begin{align}
  &  \underline{S}_{\ell}(\eta_{\ell}-\underline{\eta}_{\ell})+\overline{S}_{\ell}(\overline{\eta}_{\ell} - \eta_{\ell}) \nonumber
    \\
    &\le \alpha\begin{bmatrix}
 \underline{S}_{\ell}\overline{W}_\ell - \overline{S}_{\ell}\underline{W}_{\ell} & \overline{S}_\ell\overline{W}_\ell-\underline{S}_{\ell}\underline{W}_{\ell}
\end{bmatrix}
\begin{bmatrix}
\eta_{\ell-1}-\underline{\eta}_{\ell-1}
\\
\overline{\eta}_{\ell-1}-\eta_{\ell-1}
\end{bmatrix} \label{eq:lem_2_10} \nonumber
\end{align}
with any $\underline{S}_\ell, \overline{S}_\ell  \in \mathbb{R}^{n_L\times n_{\ell}}_+$. 

Due to (\ref{eq:lem_2_2}) which always holds with existence of $\underline{S}_{\ell}$, $\overline{S}_{\ell}$, $\ell = 0,\ldots,L$ as proved by (\ref{eq:lem_2_3}) and (\ref{eq:lem_2_4}), the above inequality ensures 
\begin{align*}
   & \underline{S}_{\ell}(\eta_{\ell}-\underline{\eta}_{\ell})+\overline{S}_{\ell}(\overline{\eta}_{\ell} - \eta_{\ell}) 
\\    
    \le& \begin{bmatrix}
 \underline{S}_{\ell-1} & \overline{S}_{\ell-1}
\end{bmatrix}
\begin{bmatrix}
\eta_{\ell-1}-\underline{\eta}_{\ell-1}
\\
\overline{\eta}_{\ell-1}-\eta_{\ell-1}
\end{bmatrix}
\\
=&\underline{S}_{\ell-1}(\eta_{\ell-1}-\underline{\eta}_{\ell-1})+\overline{S}_{\ell-1}(\overline{\eta}_{\ell-1} - \eta_{\ell-1})
\end{align*}
which can be iterated to yield
\begin{align}
    \underline{S}_{L}(\eta_{L}-\underline{\eta}_{L})+\overline{S}_{L}(\overline{\eta}_{L} - \eta_{L})  \le     \underline{S}_{0}(\eta_{0}-\underline{\eta}_{0})+\overline{S}_{0}(\overline{\eta}_{0} - \eta_{0}) . \nonumber
\end{align}

Owing to the fact of $\eta_L = \Phi(\eta_0)$, $\underline{\eta}_L = \underline{\Phi}(\underline{\eta}_0,\overline{\eta}_0)$, $\overline{\eta}_L = \overline{\Phi}(\underline{\eta}_0,\overline{\eta}_0)$ and $\underline{S}_L = \overline{S}_L = I_{n_L \times n_L}$, the following inequality can be established
\begin{equation}
    \overline{\Phi}(\underline{\eta}_0,\overline{\eta}_0) - \underline{\Phi}(\underline{\eta}_0,\overline{\eta}_0) \le     \underline{S}_{0}(\eta_{0}-\underline{\eta}_{0})+\overline{S}_{0}(\overline{\eta}_{0} - \eta_{0})
\end{equation}
for any $\underline{\eta}_0 \le \eta_0 \le \overline{\eta}_0$.
The proof is complete.
\end{proof}
\begin{remark} Lemma \ref{lemma_2} ensures the existence  of $\underline{S}_{\ell}$, $\overline{S}_{\ell}$, $\ell = 0,\ldots,L$ such that the input-output relationship in the description of (\ref{eq:lem_2_2}) holds for  auxiliary neural networks $\underline{\Phi}$ and $\overline{\Phi}$. It also provides a method to compute $\underline{S}_{\ell}$, $\overline{S}_{\ell}$, $\ell = 0,\ldots,L$, that is solving linear inequality (\ref{eq:lem_2_1}) with initialized $\underline{S}_L = \overline{S}_L = I_{n_L \times {n_L}}$. In practice, optimal solution of $\underline{S}_0$, $\overline{S}_0$ such as $\min ~\mathrm{trace(diag}\{\underline{S}_0,\overline{S}_0\})$ is of interest. With respect to objective functions of interest, optimization problems such as linear programming (LP) problems can be formulated to compute $\underline{S}_0 = \overline{S}_0$. For instance, the following LP problem can be formulated
\begin{align}
   &  \min ~\mathrm{trace(diag}\{\underline{S}_0,\overline{S}_0\})
 \nonumber   \\
    \mathrm{s.t.~} &  \alpha\begin{bmatrix}
 \underline{S}_{\ell}\overline{W}_\ell - \overline{S}_{\ell}\underline{W}_{\ell} \nonumber\\ \overline{S}_\ell\overline{W}_\ell-\underline{S}_{\ell}\underline{W}_{\ell}
\end{bmatrix} \le 
\begin{bmatrix}
\underline{S}_{\ell-1} \\ \overline{S}_{\ell-1}
\end{bmatrix},~\ell = 1,\ldots,L \nonumber
\\
& \underline{S}_L = \overline{S}_L = I_{n_L \times {n_L}} . \label{eq:LP}
\end{align}
\end{remark}

Based on Lemma \ref{lemma_2}, we are ready to derive the following result to ensure the boundedness and convergence of run-time error states, namely the practical stability of error system (\ref{eq:errorsystem}). Before presenting the result, we assume that the input vector of neural network component $\Phi(\eta_0)$ is $\eta_0 = [x, u] \in \mathbb{R}^{n_x+n_u}$. 

\begin{theorem} \label{thm3}
Consider error system (\ref{eq:errorsystem}), if there exist a diagonal matrix $X \in \mathbb{R}^{n_x \times n_x}$, matrices $\underline{Y},\overline{Y} \in \mathbb{R}^{n_x \times n_y}$ and a scalar $a \in \mathbb{R}$ such that 
\begin{align}
\label{eq:thm_3_1}
X \mathbf{1 }_{n_x \times 1} &> 0 
\\
\label{eq:thm_3_2}
XA-\underline{Y}C &> aI_{n_x \times n_x}
\\
\label{eq:thm_3_3}
XA-\overline{Y}C &> aI_{n_x \times n_x}
\\
\label{eq:thm_3_4}
\begin{bmatrix}
A^{\top}X-C^{\top}\underline{Y}^{\top} +\underline{\gamma} X+\underline{U}^{\top}X
\\
A^{\top}X-C^{\top}\overline{Y}^{\top} +\overline{\gamma} X+\overline{U}^{\top}X
 \end{bmatrix} \mathbf{1}_{n_x \times 1}& < 0
\end{align}
where $\underline{\gamma} = \underline{\gamma}_1+\underline{\gamma}_2$, $\overline{\gamma} = \overline{\gamma}_1+\overline{\gamma}_2$ and $\underline{U}$, $\overline{U}$ are defined by $\underline{S}_0 = [\underline{U},~\underline{V}]$ and $\overline{S}_0 =[\overline{U},~\overline{V}]$ where $\underline{U},\overline{U} \in \mathbb{R}^{n_x \times n_x}$, $\underline{V},\overline{V} \in \mathbb{R}^{n_x \times n_u}$, and  $\underline{S}_0,\overline{S}_0 \in \mathbb{R}^{n_x \times (n_x+n_u)}$ are the solution of the following conditions:
\begin{align}
\label{eq:thm_3_6}
\alpha\begin{bmatrix}
 \underline{S}_{\ell}\overline{W}_\ell - \overline{S}_{\ell}\underline{W}_{\ell}\\ \overline{S}_\ell\overline{W}_\ell-\underline{S}_{\ell}\underline{W}_{\ell}
\end{bmatrix} &\le 
\begin{bmatrix}
\underline{S}_{\ell-1} \\ \overline{S}_{\ell-1}
\end{bmatrix},~\ell = 1,\ldots,L
\end{align}
with $\underline{S}_L =\overline{S}_L= I_{n_x \times n_x}$, then the error system (\ref{eq:errorsystem}) is globally practically uniformly exponentially stable with observer gains $\underline{L}=X^{-1}\underline{Y}$ and $\overline{L} = X^{-1}\overline{Y}$.
\end{theorem}

\begin{proof}
First, we construct the co-positive Lyapunov function candidate $V(\underline{e},\overline{e})= \underline{V}(\underline{e})+\overline{V}(\overline{e})$, where $\underline{V}(\underline{e}) = \underline{e}^{\top}v$, $\overline{V}(\overline{e})=\overline{e}^{\top}v$ with $v = X\mathbf{1}_{n_x \times 1} \in \mathbb{R}_+^{n_x}$.

Considering $\underline{V}(\underline{e})$, one can obtain
\begin{align} \label{eq:thm3_pf_1}
    \dot{\underline{V}}(\underline{e}) 
  = \underline{e}^{\top}(A^{\top}-C^{\top}\underline{L}^{\top})v + \underline{\mathcal{F}}^{\top}v+\underline{\mathcal{G}}^{\top}v
\end{align}
where $\underline{\mathcal{F}} = f(x) - \underline{f}(\underline{x},\overline{x})$, $\underline{\mathcal{G}} = \Phi(x,u)-\underline{\Phi}(\underline{x},\overline{x},\underline{u},\overline{u})$.

Under Assumption \ref{assumption_3}, it implies that
\begin{align} \label{eq:thm3_pf_2}
    \underline{\mathcal{F}}^{\top}v \le \underline{\gamma}_1 \underline{e}^{\top}v+  \underline{\gamma}_2\overline{e}^{\top}v + \underline{\rho}^{\top}v .
\end{align}
Thus, we have
\begin{align} \label{eq:thm3_pf_3}
    \dot{\underline{V}}(\underline{e}) 
  \le \underline{e}^{\top}\underline{\Theta}v + \underline{\rho}^{\top}v+\underline{\mathcal{G}}^{\top}v
\end{align}
where $\underline{\Theta} = A^{\top}-C^{\top}\underline{L}^{\top} +(\underline{\gamma}_1 +\underline{\gamma}_2)I_{n_x \times n_x}$. 

Similarly, the following inequality can be obtained for ${\overline{V}(\overline{e})}$:
\begin{align} \label{eq:thm3_pf_4}
    \dot{\overline{V}}(\overline{e}) 
  \le \overline{e}^{\top}\overline{\Theta}v + \overline{\rho}^{\top}v+\overline{\mathcal{G}}^{\top}v
\end{align}
where $\overline{\Theta} = A^{\top}-C^{\top}\overline{L}^{\top} +(\overline{\gamma}_1 +\overline{\gamma}_2)I_{n_x \times n_x}$ and $\overline{\mathcal{G}} = \overline{\Phi}(\underline{x},\overline{x},\underline{u},\overline{u}) - \Phi(x,u)$. 

From (\ref{eq:thm3_pf_3}) and (\ref{eq:thm3_pf_4}), we have 
\begin{align} \label{eq:thm3_pf_5}
    \dot V(\underline{e},\overline{e}) 
    \le
    \begin{bmatrix}
    \underline{e}^{\top} & \overline{e}^{\top}
    \end{bmatrix}\begin{bmatrix}
    \underline{\Theta} 
    \\
     \overline{\Theta}
    \end{bmatrix}
    v + (\underline{\rho}^{\top}+\overline{\rho}^{\top})v
    + \mathcal{G}^{\top}v
\end{align}
where $\mathcal{G} = \underline{\mathcal{G}}+\overline{\mathcal{G}}$. 

Due to $\mathcal{G} = \underline{\mathcal{G}}+\overline{\mathcal{G}} = \overline{\Phi}(\underline{x},\overline{x},\underline{u},\overline{u})-\underline{\Phi}(\underline{x},\overline{x},\underline{u},\overline{u})$ and using  (\ref{eq:thm_3_6}) based on Lemma \ref{lemma_2}, it leads to
\begin{align}
    \mathcal{G} \le &\begin{bmatrix}
\underline{U} & \underline{V}
\end{bmatrix} \left(\begin{bmatrix}
x
\\
u
\end{bmatrix} - \begin{bmatrix}
\underline{x}
\\
\underline{u}
\end{bmatrix} 
\right)+
\begin{bmatrix}
\overline{U} & \overline{V}
\end{bmatrix}
\left(\begin{bmatrix}
\overline{x}
\\
\overline{u}
\end{bmatrix} - \begin{bmatrix}
{x}
\\
{u}
\end{bmatrix} 
\right) \nonumber
\\ 
= & \underline{U}\underline{e}+\underline{V}(u-\underline{u})+\overline{U}\overline{e}+\overline{V}(\overline{u}-u) \nonumber
\\
\le&
\begin{bmatrix}
\underline{U} & \overline{U}
\end{bmatrix} 
\begin{bmatrix}
\underline{e}  
\\
\overline{e}
\end{bmatrix}
+ (\underline{V}+\overline{V})(\overline{u}-\underline{u}) . \label{eq:thm3_pf_5a}
\end{align}

Therefore, one has
\begin{align} \label{eq:thm3_pf_6}
    \dot V(\underline{e},\overline{e}) \le \begin{bmatrix}
    \underline{e}^{\top} & \overline{e}^{\top}
    \end{bmatrix}\begin{bmatrix}
    \underline{\Theta} +\underline{U}^{\top} 
    \\
     \overline{\Theta}+\overline{U}^{\top}
    \end{bmatrix} v +\theta
\end{align}
where $\theta =(\overline{u}^{\top}-\underline{u}^{\top})(\underline{V}^{\top}+\overline{V}^{\top})v+(\underline{\rho}^{\top}+\overline{\rho}^{\top})v$.

Due to (\ref{eq:thm_3_4}) and $v=X\mathbf{1}_{n_x \times 1}$, it leads to
\begin{align} \label{eq:thm3_pf_7}
  \begin{bmatrix}
    \underline{\Theta} +\underline{U}^{\top} 
    \\
     \overline{\Theta}+\overline{U}^{\top}
    \end{bmatrix} v<  0
\end{align}
which implies that there always exists a sufficient small $\lambda > 0$ such that
\begin{align}
  \begin{bmatrix}
    \underline{\Theta} +\underline{U}^{\top} 
    \\
     \overline{\Theta}+\overline{U}^{\top}
    \end{bmatrix} v <- \lambda 
    \begin{bmatrix}
    I_{n_x \times n_x}
    \\
    I_{n_x \times n_x} 
    \end{bmatrix} v
\end{align}
and that implies
\begin{align}\label{eq:thm3_pf_8}
    \dot V(\underline{e},\overline{e}) \le -\lambda \underline{e}^{\top} v -\lambda \overline{e}^{\top} v + \theta =     - \lambda V(\underline{e},\overline{e}) + \theta .
\end{align}

Defining $\underline{v}$ and $\overline{v}$ the minimal and maximal element in $v$, (\ref{eq:thm3_pf_8}) implies that
\begin{align*}
   \underline{v}\left\|\xi\right\|  \le e^{-\lambda(t-t_0)} \overline{v} \left\|\xi_0\right\|+ \frac{\theta}{\lambda}
    \Rightarrow  \left\|\xi\right\|  \le Ce^{-\lambda(t-t_0)}  \left\|\xi_0\right\|+r
\end{align*}
where $\xi^{\top} = [\underline{e}^{\top}, \overline{e}]^{\top}$, $C= {\overline{v}}\slash{\underline{v}}$ and $r =  {\theta}\slash{\lambda \underline{v}}$. Therefore, the error system is globally practically uniformly exponentially stable, the error state converges to  ball $\mathcal{B}_r = \{x \in \mathbb{R}^{n_x} \mid \left\|x\right\| \ge r, r \ge 0\}$  exponentially at a decay rate of $\lambda$. The proof is complete. 
\end{proof}

The design process of observer gains $\underline{L}$ and $\overline{L}$ allows one to use coordinate transformation techniques used in several works such as  \cite{efimov2016design,raissi2011interval,li2019interval} to relax the conditions of both Mezleter and Hurwitz conditions being satisfied.  Based on Theorem \ref{thm3}, an design algorithm is proposed in Algorithm \ref{alg1}. The outputs of the algorithm, the auxiliary neural networks $\underline{\Phi}$, $\overline{\Phi}$ and observer gains $\underline{L}$, $\overline{L}$, are able to ensure the run-time boundedness as well as the convergence of the error state in terms of practical stability. 

\begin{algorithm}[ht!]
\SetAlgoLined
\SetKwInOut{Input}{Input}
\SetKwInOut{Output}{Output}
\SetKw{Return}{return}
Compute $\underline{W}_\ell$, $\overline{W}_\ell$, $\ell = 1,\ldots,L$ by (\ref{eq:W_1}), (\ref{eq:W_2}) and obtain neural networks $\underline{\Phi}$, $\overline{\Phi}$\;
Solve LP problem (\ref{eq:LP}) to obtain $\underline{S}_0$ and $\overline{S}_0$ \;
Compute $\underline{U}$, $\overline{U}$ by $\underline{S}_0 = [\underline{U},~\underline{V}]$, $\overline{S}_0 =[\overline{U},~\overline{V}]$ where $\underline{U},\overline{U} \in \mathbb{R}^{n_x \times n_x}$, $\underline{V},\overline{V} \in \mathbb{R}^{n_x \times n_u}$ \;
Solve LP problem (\ref{eq:thm_3_1})-(\ref{eq:thm_3_4}) to obtain  $X$, $\underline{Y}$, $\overline{Y}$ \;
Compute observer gains $\underline{L}$, $\overline{L}$. 
\caption{Run-Time Safety State Estimator Design} \label{alg1}
\end{algorithm}

A numerical example is proposed to illustrate the design process of Algorithm \ref{alg1}. 
\begin{example}
Consider a neural-network-enabled system in the form of $\dot x = Ax+\Phi(x,u)$, $y =Cx$, where system matrices are
\begin{align*}
    A = \begin{bmatrix}
    -2 & 1
    \\
    3 & -5
    \end{bmatrix},~C = \begin{bmatrix}
    0 & 1
    \end{bmatrix}
\end{align*}
and neural network $\Phi$ is determined by 
\begin{align*}
    & W_1 = \begin{bmatrix}
    0.6266  &   0.8433 &   0.3241 \\
   -0.2485  & -1.5838  & -0.5620 \\
    0.5243  & -1.4939  &  1.1992 \\
   -0.4300  & -1.4659  &  0.1102 \\
    0.2629  &  0.6789 &  -1.2695
    \end{bmatrix},~
    b_1 = \begin{bmatrix}
    -1.0191 \\
   -1.3852 \\
    0.9549 \\
   -0.6011 \\
   -1.1719 
    \end{bmatrix}
    \\
     &   W_2^{\top} = \begin{bmatrix}
        -0.4617 &   -0.6691  \\
    0.6824 &  0.3819  \\
    0.2419  &  0.3326  \\
    0.0344 &  -0.7591  \\
    0.4333 &  -0.6569
    \end{bmatrix},~
    b_2 = \begin{bmatrix}
     -1.0719
     \\
   -1.0741
    \end{bmatrix} 
\end{align*}
and activation functions are \texttt{tanh} and \texttt{purelin}. 

\emph{Step 1. Design Auxiliary Neural Networks:} By (\ref{eq:W_1}) and (\ref{eq:W_2}), matrices $\underline{W}_\ell$, $\overline{W}_\ell$, $\ell = 1, 2$ are as follows:
\begin{align*}
    &\underline{W}_1 = \begin{bmatrix}
            0    &     0     &    0  \\
   -0.2485 & -1.5838 &  -0.5620 \\
         0 &  -1.4939     &   0  \\
   -0.4300 &  -1.4659    &     0 \\
         0 &        0  & -1.2695
    \end{bmatrix}
    \\
    &\overline{W}_1 = \begin{bmatrix}
      0.6266 &    0.8433  &   0.3241 \\
         0   &       0   &    0 \\
    0.5243    &      0    1.1992 \\
         0  &        0  &   0.1102\\
    0.2629 &    0.6789  &        0
        \end{bmatrix}
        \\
       & \underline{W}_2 = 
        \begin{bmatrix}
   -0.4617    &     0  &      0       &  0     &    0
   \\
   -0.6691    &     0   &      0  & -0.7591  & -0.6569
        \end{bmatrix}
        \\
          & \overline{W}_2 = 
        \begin{bmatrix}
           0  & 0.6824 &   0.2419  &  0.0344 &   0.4333 \\
         0  &  0.3819  &  0.3326    &     0   &      0
        \end{bmatrix} .
\end{align*}

\emph{Step 2: Design Observer Gains:} By (\ref{eq:thm_3_1})--(\ref{eq:thm_3_6}), the observer gains are computed as 
\begin{align*}
    \underline{L} = \begin{bmatrix}
    0 \\
   12.0394
    \end{bmatrix},~
        \overline{L} = \begin{bmatrix}
    1 \\
   8.0044
    \end{bmatrix} .
\end{align*}

Assuming input $u = 10\sin(5t)$, thus we have $\underline{u} = -10$ and $\overline{u} = 10$. The initial state is assumed to be bounded in $[-1,1]$. The run-time safety monitoring of system state $x_1(t)$ and $x_2(t)$ is illustrated in Figure \ref{fig1}. The run-time state trajectories $x_1(t)$, $x_2(t)$ are bounded in run-time estimated states $\underline{x}_i(t)$, $\overline{x}_i(t)$, $i =1,2$, in other words, the safety of system state $x(t)$ can be monitored by $\underline{x}_i(t)$, $\overline{x}_i(t)$, $i =1,2$ in run time.  
\begin{figure}\label{fig1}
\centering
	\includegraphics[width=9.5cm]{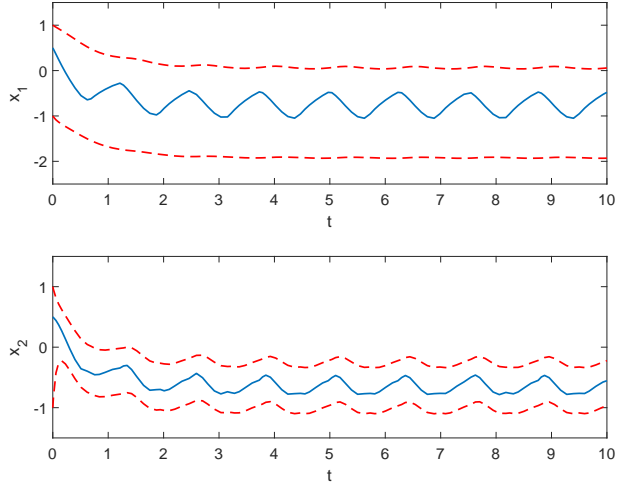}
	\caption{State response of $x(t)$ (solid lines) and run-time safety monitoring of $\underline{x}(t)$ and $\overline{x}(t)$ (dashed lines). State response $x(t)$ is bounded between states $\underline{x}(t)$, $\overline{x}(t)$ of state estimator in run time.}
	\label{x} 
\end{figure}
\end{example}

As one of the most common neural-network-enabled dynamical systems, the neural network is driven by the measurement of the system, which means the input of the neural network is the measurement of the system $y(t)$ instead of system state $x(t)$. For instance, neural network control systems use the measurement $y(t)$ to compute the control input instead of system state $x(t)$ since system state $x(t)$ may not be measurable. This class of systems with neural network $\Phi(y,u)$ are in the following description of
\begin{align}\label{eq:Lipsystem_y}
\mathfrak{L}_y:
\begin{cases}
    \dot x  = Ax + f(x) + \Phi(y,u)
    \\
    y = Cx
\end{cases}
\end{align}
where neural network $\Phi(y,u)$ is measured output driven. Since the output $y(t)$ is measurable in run time, it can be employed in the safety monitoring. The run-time safety state estimator is developed in the form of
\begin{align}
\label{eq:observer_y}
\mathfrak{E}_y:
\begin{cases}
    \dot{\underline{x}}  = A\underline{x} + \underline{f}(\underline{x},\overline{x}) + \underline{\Phi}(y,\underline{u},\overline{u})+\underline{L}(y - C\underline{x})
    \\
    \dot{\overline{x}}  = A\overline{x} + \overline{f}(\underline{x},\overline{x}) + \overline{\Phi}(y,\underline{u},\overline{u})+\overline{L}(y - C\overline{x})
\end{cases}
\end{align}
where $\underline{\Phi}$ and $\overline{\Phi}$ are defined by (\ref{eq:nn_1}) and (\ref{eq:nn_2}) with $\underline{y}=\overline{y}=y$, respectively. Consequently, the error dynamics is in the form of 
\begin{align}
\label{eq:errorsystem_y}
\begin{cases}
    \dot{\underline{e}}  = (A-\underline{L}C)\underline{e} + f(x) - \underline{f}(\underline{x},\overline{x}) + \Phi(y,u)-\underline{\Phi}(y,\underline{u},\overline{u})
    \\
    \dot{\overline{e}}  = (A-\overline{L}C)\overline{e} + \overline{f}(\underline{x},\overline{x}) - f(x) + \overline{\Phi}(y,\underline{u},\overline{u})-\Phi(y,u)
\end{cases}
\end{align}

The following result represents the observer gain design process in (\ref{eq:observer_y}). 
\begin{corollary} \label{cor_1}
Consider error system (\ref{eq:errorsystem_y}), if there exist a diagonal matrix $X \in \mathbb{R}^{n_x \times n_x}$, matrices $\underline{Y},\overline{Y} \in \mathbb{R}^{n_x \times n_y}$ and a scalar $a \in \mathbb{R}$ such that 
\begin{align}
\label{eq:cor_1_1}
X \mathbf{1 }_{n_x \times 1} &> 0 
\\
\label{eq:cor_1_2}
XA-\underline{Y}C &> aI_{n_x \times n_x}
\\
\label{eq:cor_1_3}
XA-\overline{Y}C &> aI_{n_x \times n_x}
\\
\label{eq:cor_1_4}
\begin{bmatrix}
A^{\top}X-C^{\top}\underline{Y}^{\top} +\underline{\gamma} X
\\
A^{\top}X-C^{\top}\overline{Y}^{\top} +\overline{\gamma} X
 \end{bmatrix} \mathbf{1}_{n_x \times 1}& < 0
\end{align}
where $\underline{\gamma} = \underline{\gamma}_1+\underline{\gamma}_2$, $\overline{\gamma} = \overline{\gamma}_1+\overline{\gamma}_2$, then the error system (\ref{eq:errorsystem_y}) is globally practically uniformly exponentially stable with observer gains $\underline{L}=X^{-1}\underline{Y}$ and $\overline{L} = X^{-1}\overline{Y}$.
\end{corollary}
\begin{proof}
Construct a co-positive Lyapunov function candidate $V(\underline{e},\overline{e}) = \underline{e}^{\top}v+\overline{e}^{\top}v$ where $v = X\mathbf{1}_{n_x \times 1} \in \mathbb{R}^{n_x}_+$. Following the same guideline in Theorem \ref{thm3}, the following inequality can be obtained
\begin{align} \label{eq:cor_1_pf_5}
    \dot V(\underline{e},\overline{e}) 
    \le
    \begin{bmatrix}
    \underline{e}^{\top} & \overline{e}^{\top}
    \end{bmatrix}\begin{bmatrix}
    \underline{\Theta} 
    \\
     \overline{\Theta}
    \end{bmatrix}
    v + (\underline{\rho}^{\top}+\overline{\rho}^{\top})v
    + \mathcal{G}^{\top}v
\end{align}
where $\overline{\Theta} = A^{\top}-C^{\top}\overline{L}^{\top} +(\overline{\gamma}_1 +\overline{\gamma}_2)I_{n_x \times n_x}$, $\underline{\Theta} = A^{\top}-C^{\top}\underline{L}^{\top} +(\underline{\gamma}_1 +\underline{\gamma}_2)I_{n_x \times n_x}$, and $\mathcal{G}  = \overline{\Phi}(y,\underline{u},\overline{u})-\underline{\Phi}(y,\underline{u},\overline{u})$. 

Based on Lemma \ref{lemma_2} and using the fact of $\underline{y} = \overline{y} =y$ in $\underline{\Phi}$ and $\overline{\Phi}$, it leads to
\begin{align}
    \mathcal{G} \le &\begin{bmatrix}
\underline{U} & \underline{V}
\end{bmatrix} \left(\begin{bmatrix}
y
\\
u
\end{bmatrix} - \begin{bmatrix}
y
\\
\underline{u}
\end{bmatrix} 
\right)+
\begin{bmatrix}
\overline{U} & \overline{V}
\end{bmatrix}
\left(\begin{bmatrix}
y
\\
\overline{u}
\end{bmatrix} - \begin{bmatrix}
y
\\
{u}
\end{bmatrix} 
\right) \nonumber
\\
\le &
\begin{bmatrix}
\underline{U} & \overline{U}
\end{bmatrix} 
\begin{bmatrix}
0
\\
0
\end{bmatrix}
+ (\underline{V}+\overline{V})(\overline{u}-\underline{u})   \nonumber
\\
= & (\underline{V}+\overline{V})(\overline{u}-\underline{u}) 
\label{eq:cor_1_pf_5a}
\end{align}
which is irrelevant to $\underline{U}$ and $\overline{U}$. Then, we have
\begin{align} \label{eq:cor_1_pf_6}
    \dot V(\underline{e},\overline{e}) \le \begin{bmatrix}
    \underline{e}^{\top} & \overline{e}^{\top}
    \end{bmatrix}\begin{bmatrix}
    \underline{\Theta}  
    \\
     \overline{\Theta}
    \end{bmatrix} v +\theta
\end{align}
where $\theta =(\overline{u}^{\top}-\underline{u}^{\top})(\underline{V}^{\top}+\overline{V}^{\top})v+(\underline{\rho}^{\top}+\overline{\rho}^{\top})v$.

Due to (\ref{eq:cor_1_4}) and following the same guidelines in Theorem \ref{cor_1}, the following inequality can be derived
\begin{align}\label{eq:cor_1_pf_8}
    \dot V(\underline{e},\overline{e}) \le -\lambda \underline{e}^{\top} v -\lambda \overline{e}^{\top} v + \theta
    = 
    - \lambda V(\underline{e},\overline{e}) + \theta
\end{align}
which ensure the practical stability of error system. The proof is complete. 
\end{proof}
\begin{remark}
As Corollary \ref{cor_1} indicates, the design process of observer gain computation has nothing to do with the neural network. The observer gains $\underline{L}$ and $\overline{L}$ are obtained by solving an LP problem in terms of (\ref{eq:cor_1_1})--(\ref{eq:cor_1_4}) which is dependent upon system dynamics without considering neural network components. This is because the measurable output $y$ makes the portion of the output of neural network $\Phi$ driven by $y$ completely compensated by the outputs of auxiliary neural networks $\underline{\Phi}$, $\overline{\Phi}$ which are also driven by same values of  measurement $y$. This promising feature of irrelevance to neural networks leads this developed methods to be able to deal with dynamical systems with large-scale neural network components such as deep neural network controllers regardless of the size of neural networks.
\end{remark}

\section{Application to Adaptive Cruise Control Systems}
\begin{figure}
    \includegraphics[width=8.8cm]{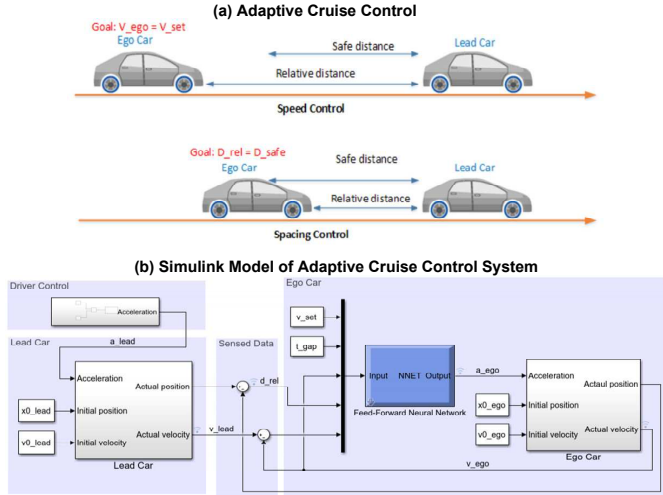}
    \caption{\boldmath  Illustration of adaptive cruise control systems and simulink block diagram of the closed-loop system.}
    \label{acc_sys} 
\end{figure}

In this section, the developed run-time safety monitoring approach will be evaluated  by an adaptive cruise control (ACC) system which is under control of a neural network controller as
depicted in Figure \ref{acc_sys}. Two cars are involved in the ACC system, an ego car with ACC module and a lead car. A radar sensor is equipped on the car to measure
the distance to the lead car in run time. The run-time measured distance is denoted by $d_{\mathrm{rel}}$. Moreover, the relative velocity
against the lead car is also measured in run time which is denoted by $v_{\mathrm{rel}}$. There are two system operating modes including speed control and spacing control. Two control modes are operating in run time. In speed control mode, the ego car travels at a speed of $v_{\mathrm{set}}$ set by the driver. In spacing control mode, the ego car has to maintain a safe distance from the lead car denoted by $d_{\mathrm{safe}}$. The system dynamics of ACC is expressed in the following form of
\begin{align}\label{acc}
\left\{ {\begin{array}{*{20}l}
    \dot x_{l}(t) = v_{l}(t)\\
    \dot v_{l}(t) = \gamma_{l}(t)\\
    \dot \gamma_{l}(t) = -2\gamma_{l}(t) + 2 \alpha_{l}(t) - \mu v^2_{l}(t)
    \\
    \dot x_{e} = v_{e}(t)
    \\
    \dot v_{e}(t) = \gamma_e(t)
    \\
    \dot \gamma_{e}(t) = -2\gamma_{e}(t)+2\alpha_{e}(t)-\mu v^{2}_{e}(t)
    \end{array} } \right.
\end{align}
where $x_l$, $v_l$ and $\gamma_l$ are the position, velocity and actual acceleration of the lead car, and $x_e$, $v_e$ and $\gamma_e$ are the position, velocity and actual acceleration of the ego car, respectively.  $\alpha_l$ and $\alpha_e$ is the acceleration
control inputs applied to the lead and ego car. $\mu = 0.0001$ is the friction parameter. A $2 \times 20$ feedforward neural network controller with $\texttt{tanh}$ and $\texttt{purelin}$ is trained for the ACC system. Specifically, the measurement of the ACC system which also performs as the inputs to the neural network ACC control module are listed in Table \ref{tab1}.
\begin{table}[h!]
\centering
\caption{Measured Outputs of ACC system and Inputs to Neural Network Controller}
\label{tab1}
\begin{tabular}{ |l||l| }
\hline
Driver-set velocity & $v_{\mathrm{set}}$  \\  
 \hline
Time gap & $t_{\mathrm{gap}}$ \\  
 \hline
Velocity of the ego car & $v_e$ \\  
 \hline
Relative distance to the lead car & $d_{\mathrm{rel}} = x_l - x_e$\\ 
 \hline
 Relative velocity to the lead car & $v_{\mathrm{rel}} = v_l - v_e$\\ 
 \hline
\end{tabular}
\end{table}

The output for the  neural network ACC controller is the acceleration of the ego car, namely $\alpha_e$. In summary, the neural network controller for the acceleration control of the ego car is in the form of
\begin{align}\label{sam_acc}
\alpha_{e}(t) = \Phi( d_{\mathrm{rel}}(t),v_{\mathrm{rel}}(t),v_{e}(t),v_{\mathrm{set}}(t),t_{\mathrm{gap}}).
\end{align}

Letting $x = [x_l,v_l,\gamma_l,x_e,v_e,\gamma_e]^{\top}$, $u = [\alpha_l,v_{\mathrm{set}},t_{\mathrm{gap}}]^{\top}$ and $y = [v_e,d_{\mathrm{rel}},v_{\mathrm{rel}}]^{\top}$, the ACC system can be rewritten in the following neural-network-enabled system
\begin{align}
\begin{cases}
    \dot x = Ax+f(y)+\tilde \Phi(y,u)
    \\
    y = Cx
\end{cases}
\end{align}
where system matrices $A$, $C$, nonlinearity $f(y)$ and neural network component $\Phi(y,u)$ are defined as below:
\begin{align*}
    A &= 
     \begin{bmatrix}
     0 & 1 & 0 & 0 & 0 & 0 
     \\
     0 & 0 & 1 & 0 & 0 & 0 
     \\
     0 & 0 & -2 & 0 & 0 & 0 
     \\
     0 & 0 & 0 & 0 & 1 & 0 
     \\
     0 & 0 & 0 & 0 & 0 & 1
     \\
     0 & 0 & 0 & 0 & 0 & -2
    \end{bmatrix}
    \\
    C &= \begin{bmatrix}
    0 & 0 & 0 & 0 & 1 & 0 \\
     1 & 0 & 0 & -1 & 0 & 0 \\
     0 & 1 & 0 & 0 & -1 & 0
    \end{bmatrix}
    \\
    f(y) &= \begin{bmatrix}0 & 0 & -0.0001 (y_1+y_3)^2 & 0 & -0.0001 y_1^2
    \end{bmatrix}^{\top}
    \\
    \tilde \Phi(y,u) &= 
    \begin{bmatrix}
    0 & 0 & 2u_1 & 0 &0 &  \Phi(y,u_2,u_3) 
    \end{bmatrix}^{\top} 
\end{align*}
where $\Phi(y,u_2,u_3)$ is the neural network controller (\ref{sam_acc}). 
In addition, considering the physical limitations of the vehicle dynamics, the acceleration is constrained to the range $[-3,2]$ ($m/s^2$), thus input $u_1 \in [-3,2]$. 

\weiming{Since the nonlinearity of $f(y)$ can be obtained with the measurement of $y$, we can let $f(y) = \underline{f}(\underline{x},\overline{x}) = \overline{f}(\underline{x},\overline{x})$ in state estimator (\ref{eq:observer_y}),} which is thus constructed as follows
\begin{align}
\label{}
\begin{cases}
    \dot{\underline{x}}  = A\underline{x} + f(y) + \underline{\Phi}(y,\underline{u},\overline{u})+\underline{L}(y - C\underline{x})
    \\
    \dot{\overline{x}}  = A\overline{x} +f(y) + \overline{\Phi}(y,\underline{u},\overline{u})+\overline{L}(y - C\overline{x})
\end{cases} .
\end{align}

The auxiliary neural networks $\underline{\Phi}$ and $\overline{\Phi}$ are designed based on $\Phi$ according to (\ref{eq:nn_1}) and (\ref{eq:nn_2}). The observer gains $\underline{L}$ and $\overline{L}$ can be computed by Corollary \ref{cor_1} via solving a collection of LP problems. 

\begin{figure}
    \includegraphics[width=9.5cm]{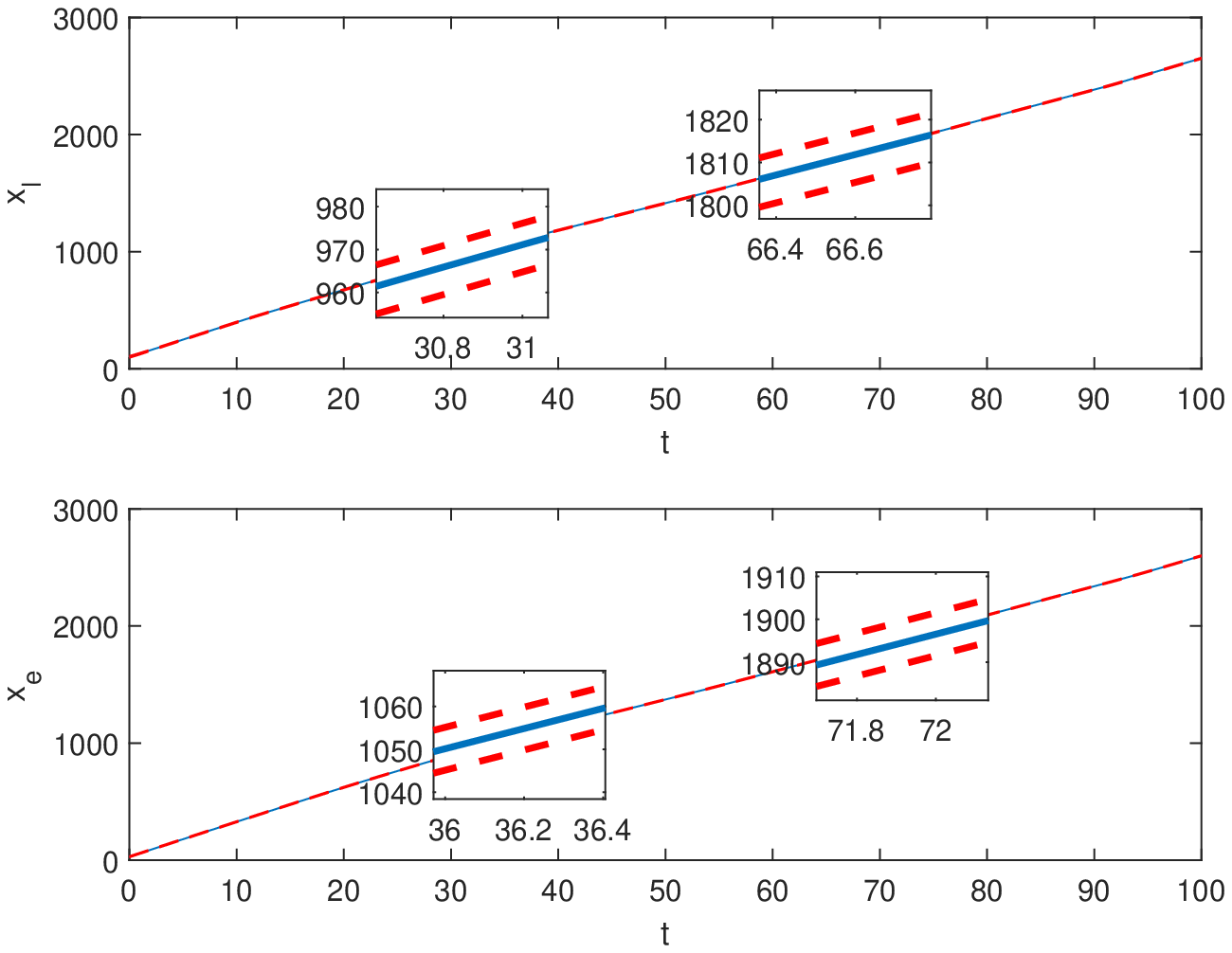}
    \caption{Run-time safety monitoring of positions $x_l(t)$ and $x_e(t)$. The state trajectories $x(t)$ (solid lines) are bounded within the estimated bounds $\underline{x}(t)$, $\overline{x}(t)$ (dashed lines). Magnified time windows are used for clear clarification. }
    \label{acc_x} 
\end{figure}
\begin{figure}
    \includegraphics[width=9.5cm]{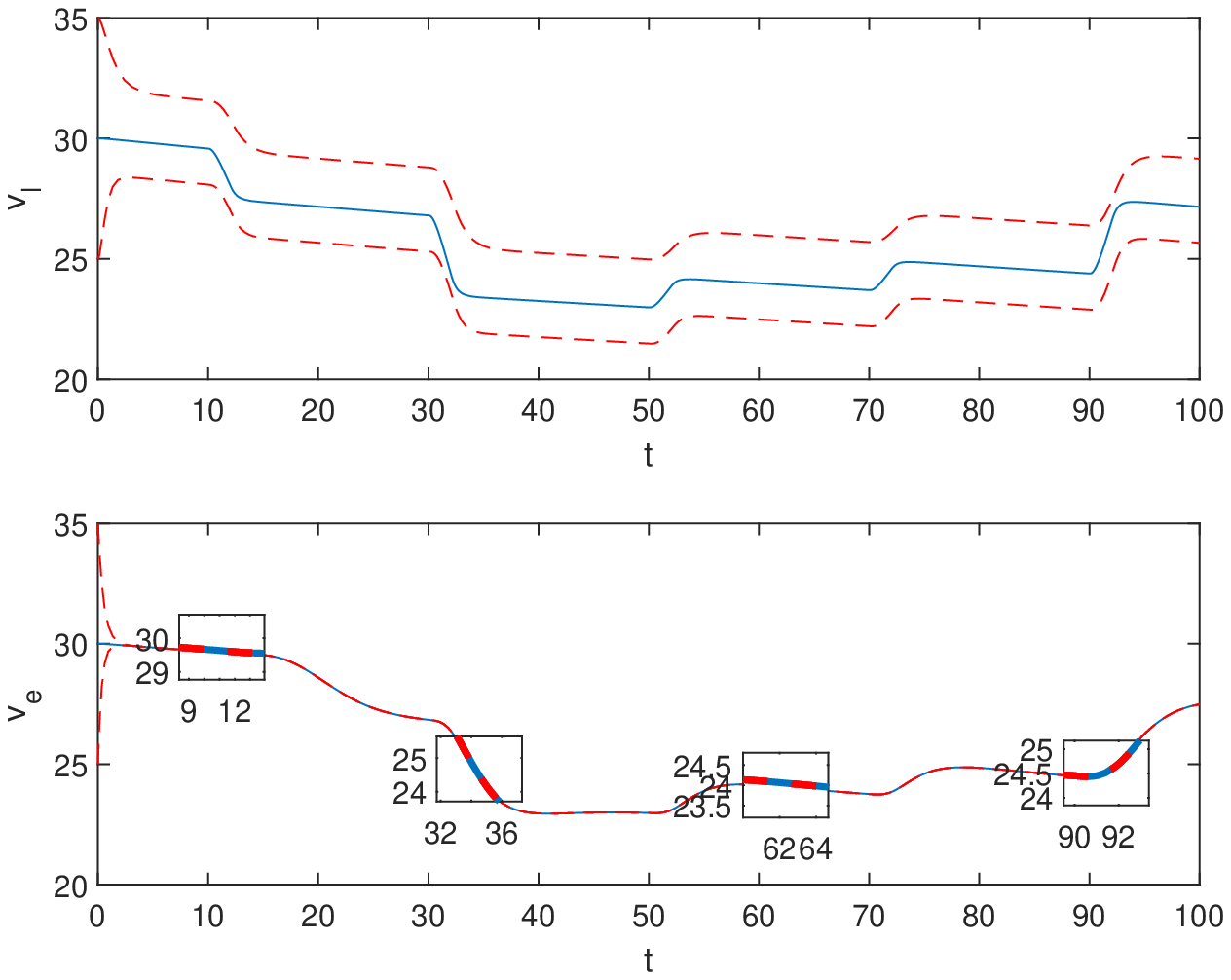}
    \caption{Run-time safety monitoring of velocities $v_l(t)$ and $v_e(t)$. The state trajectories $v(t)$ (solid lines) are bounded within the estimated bounds $\underline{v}(t)$, $\overline{v}(t)$ (dashed lines). Magnified time windows are used for clear clarification.}
    \label{acc_v} 
\end{figure}
\begin{figure}
    \includegraphics[width=9.5cm]{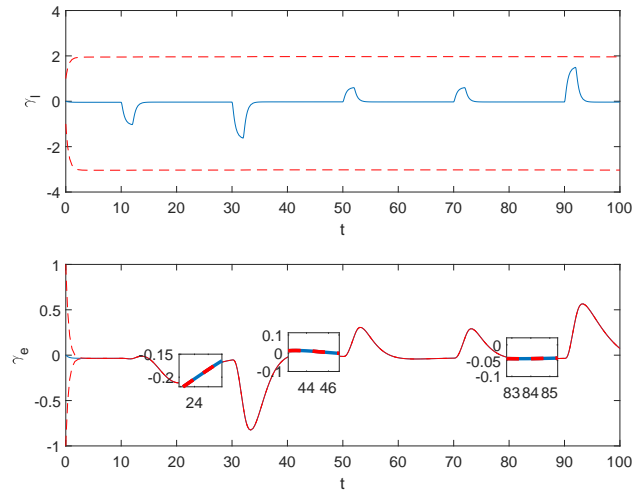}
    \caption{Run-time safety monitoring of accelerations $\gamma_l(t)$ and $\gamma_e(t)$. The state trajectories $\gamma(t)$ (solid lines) are bounded within the estimated bounds $\underline{\gamma}(t)$, $\overline{\gamma}(t)$ (dashed lines). Magnified time windows are used for clear clarification.}
    \label{acc_a} 
\end{figure}

The run-time boundary estimations of state trajectories of positions $\{x_l(t),x_e(t)\}$, velocities $\{v_l(t),v_e(t)\}$ and accelerations $\{\gamma_l(t),\gamma_e(t)\}$ during ACC system evolves in time interval $[0,100]$ are shown in Figures \ref{acc_x}, \ref{acc_v} and \ref{acc_a}. As shown in these simulation results, the state trajectories are always bounded within the lower  and  upper-bounds of observers which can be used as a run-time safety monitoring state for  system state during operation.

\section{Conclusions}
The run-time safety monitoring problem of neural-network-enabled dynamical systems is addressed in this paper. The online lower- and upper-bounds of state trajectories can be provided by the run-time safety estimator in the form of interval Luenberger observer form. The design process includes two essential parts, namely two auxiliary neural networks and observer gains. In summary, two auxiliary neural networks are derived from the neural network component embedded in the original dynamical system and observer gains are computed by a family of LP problems. Regarding neural networks driven by measurements of the system, it is noted that the design process is independent with the neural network so that there is no scalability concern for the size of neural networks. An application to ACC system is presented to validate the developed method. Further applications to complex dynamical systems such as systems with switching behaviors \cite{zhu2019quasi,xiang2016necessary,zhang2016mode,xiang2017output,xiang2012robust} should be considered in the future. Beyond the run-time safety monitoring approach developed in this paper, the future work should be the run-time correction of neural networks once the unsafe behavior is detected. Moreover, 
run-time safety-oriented training of neural networks such as online neural network controller training with safety guarantees should be also considered based on the techniques developed in this paper.

\bibliographystyle{IEEEtran}

\bibliography{ref}

\end{document}